\documentclass[conference]{IEEEtran}

\usepackage[utf8]{inputenc} 
\usepackage[T1]{fontenc}
\usepackage{url}
\usepackage{ifthen}
\usepackage{cite}
\usepackage[cmex10]{amsmath} 
\usepackage{amssymb,amsthm}
\usepackage{subcaption}
\usepackage{tikz}
\usepackage[noend]{algpseudocode}
\usepackage{algorithm}
\newtheorem{prop}{Proposition}
\newtheorem{defn}{Definition}

\newtheorem{ex}{Example}

\DeclareMathOperator{\wt}{wt}

\DeclareMathOperator{\lcm}{lcm}

\DeclareMathOperator{\aut}{Aut}
\DeclareMathOperator{\ev}{ev}
\DeclareMathOperator{\triv}{Triv}

\begin{document}
\title{On the dependency between the code symmetries and the decoding efficiency} 
 \author{%
	\IEEEauthorblockN{Kirill Ivanov and R\"udiger Urbanke}
	\IEEEauthorblockA{School of Computer and Communication Sciences\\
		EPFL, CH-1015 Lausanne, Switzerland\\
		Email: \{kirill.ivanov, ruediger.urbanke\}@epfl.ch}
}

\maketitle
\begin{abstract}
	A framework of monomial codes is considered, which includes linear codes generated by the evaluation of certain monomials. Polar and Reed-Muller codes are the two best-known representatives of such codes and can be considered as two extreme cases. Reed-Muller codes have a large automorphism group but their low-complexity maximum likelihood decoding still remains an open problem. On the other hand, polar codes have much less symmetries but admit the efficient near-ML decoding.
	
	We study the dependency between the code symmetries and the decoding efficiency. We introduce a new family of codes, partially symmetric monomial codes. These codes have a smaller group of symmetries than the Reed-Muller codes and are in this sense "between" RM and polar codes. A lower bound on their parameters is introduced along with the explicit construction which achieves it. Structural properties of these codes are demonstrated and it is shown that they often have a recursive structure. 
\end{abstract}
\section{Introduction}
Reed-Muller (RM) codes \cite{muller1954rm,reed1954rm} are very well known in the classical coding theory. It was recently proved that this family of codes achieves the capacity of the binary erasure channel under ML decoding \cite{Kudekar2016RM}. The equivalent question for general BMS channels is still open and so is low-complexity near-ML decoding of RM codes.

Reed-Muller codes have rich structural properties, such as invariance w.r.t. a large group of permutations that can be expressed as actions on monomials. The best decoding algorithms for RM codes exploit this invariance and in particular the fact that RM code can be represented as Plotkin concatenation of two smaller RM codes in several different ways. The projection-aggregation algorithm \cite{ye2019recursive} approaches ML performance for short and moderate lengths but its complexity grows exponentially with the code order. This makes it feasible only for low-rate codes. Dumer and Shabunov proposed \cite{dumer2006rec} to perform the decomposition recursively for different permutations and consider a list of most likely codewords at each stage. This approach and its variations achieves ML performance with the list size that grows exponentially with a code length. The use of multiple permutations brings the additional benefit of parallelism, which might lead to latency reduction and simplification of the hardware implementation \cite{kamenev2019rm}.

Polar codes \cite{Arikan2009polar} achieve the capacity of an arbitrary BMS channel. Contrary to RM codes, they are specifically constructed so that successive cancellation list (SCL) decoding with small list size \cite{Tal2015list} (which works similarly to Dumer-Shabunov list decoder) is sufficient for near-ML performance. However, this property comes at the cost of the code structure, so both methods designed for RM codes are inefficient for polar codes and polar-like constructions with better finite-length performance, such as CRC-aided polar codes \cite{Tal2015list} or polar subcodes \cite{trifonov2016subcodes}.

One can now ask whether it would be possible to find codes that have a smaller group of symmetries than Reed-Muller codes but also require a smaller decoding complexity for near-ML performance. Previous works include an efficient construction for two permutations \cite{wang2019gn}. It performs worse compared to polar codes under list decoding but allows turbo-like decoding with significantly smaller latency. Polar codes with explicit permutation group are proposed in \cite{kamenev2019polar}, which perform worse compared to polar codes.

In this paper, we continue this line of research and investigate this question from a code structure point of view. Here we consider codes that can be obtained via evaluations of monomials \cite{bardet2016algebraic}. Polar and Reed-Muller codes can both be described in this framework. We introduce a family of codes with certain symmetries and show the lower bound on their parameters. A channel-independent construction achieving this bound is proposed and it is shown that in some cases the obtained codes have the recursive structure. \footnote{The source code to reproduce the bounds and code constructions is available at https://github.com/kir1994/PartSymCodes}
\section{Monomial codes}

\subsection{Monomial codes}
Let $\{ x_1, \ldots, x_m \}$ be a collection of $m$ variables taking their values in $\mathbb F_2$, let $v = (v_1, \ldots, v_m) \in \mathbb{F}_2^m$ be any binary $m$-tuple, and let $\wt(\cdot)$ denote Hamming weight. Then,
\[
x^v = \prod_{i=1}^{m} x_i^{v_i}
\]
denotes a monomial of degree $\wt(v)$. Its evaluation vector $\ev(x^v)\in \mathbb F_2^m$ can be obtained by evaluating $x^v$ at all points of
$\mathbb{F}_2^m$. It can be noticed that $\wt(\ev(x^v))=2^{m-\wt(v)}$.

Condider the set of all $m$-variate monomials $$M_m=\{x^v | v \in \mathbb F_2^m \}.$$ Any its subset $M_{\mathcal C}$ is a \textit{generating set} of a $(2^m, |M_{\mathcal C}|)$ monomial code $\mathcal C$ spanned by the evaluations of these monomials, i.e., 

$$
\mathcal C = \text{span}(\{ \ev(x^v) |\  x^v \in M_{\mathcal C} \}).
$$

Minimum distance of a monomial code $\mathcal C$ is $2^{m-r^+(\mathcal C)}$ \cite{bardet2016algebraic}, where
\begin{equation}
r^+(\mathcal C)=\max_{x^v \in M_{\mathcal C}}\wt (v).
\end{equation}

\subsection{Polar codes}
A $(n=2^m,k)$ polar code with the set of frozen symbols $\mathcal
F$ is a binary linear block code generated by
rows with indices $i\in[n]\setminus \mathcal F$ of the matrix
$A_m=A^{\otimes m}$, where $A=\begin{pmatrix}1&0\\1&1\end{pmatrix}$
and $\otimes m$ denotes $m$-times Kronecker product of a matrix
with itself. For a given BMS channel $W$, the set $\mathcal F$ contains indices of the least reliable bit subchannels under successive cancellation decoding.

It can be seen that the $i$th row of $A_m$ is in fact the evaluation vector of monomial $x^{\overline i}$, where $\overline i$ is the bitwise NOT operation. Hence, polar codes are monomial codes with generating set $M_{\mathcal F}=\{x^{\overline i} | i \in \mathcal F^c\}$.

\subsection{Reed-Muller codes}
The Reed-Muller code RM$(r,m)$ is spanned by evaluations of $m$-variate monomials of degree not exceeding $r$. RM$(r,m)$ code has length $2^m$, dimension $\sum_{i=0}^{r}\binom{m}{i}$ and minimum distance $2^{m-r}$. By definition, it is a monomial code with generating set $M_{r,m}=\{x^v | \wt(v)\le r\}$.

\section{Projections}
\subsection{Factor graph permutations}
\begin{figure}
	\centering
	\scalebox{1.25}{\begin{tikzpicture}[thick,scale=0.55, every node/.style={scale=0.6}]
		\input{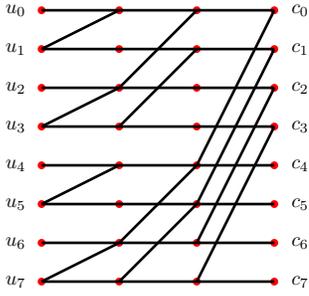}
		
		\draw node at (    5,    0) {{\color{black}$c_0$}};
		\draw node at (	5,		-0.75)  {{\color{black}$c_1$}};
		\draw node at (	5,		-1.5)  {{\color{black}$c_2$}};
		\draw node at (	5,		-2.25)  {{\color{black}$c_3$}};
		\draw node at (	5,		-3)  {{\color{black}$c_4$}};
		\draw node at (	5,		-3.75)  {{\color{black}$c_5$}};
		\draw node at (	5,		-4.5)  {{\color{black}$c_6$}};
		\draw node at (	5,		-5.25)  {{\color{black}$c_7$}};
		
		\draw node at (    -0.5,    0) {{\color{black}$u_0$}};
		\draw node at (	-0.5,		-0.75)  {{\color{black}$u_1$}};
		\draw node at (	-0.5,		-1.5)  {{\color{black}$u_2$}};
		\draw node at (	-0.5,		-2.25)  {{\color{black}$u_3$}};
		\draw node at (	-0.5,		-3)  {{\color{black}$u_4$}};
		\draw node at (	-0.5,		-3.75)  {{\color{black}$u_5$}};
		\draw node at (	-0.5,		-4.5)  {{\color{black}$u_6$}};
		\draw node at (	-0.5,		-5.25)  {{\color{black}$u_7$}};
\end{tikzpicture}}
	\caption{Polar factor graph for $m=3$.}
	\label{fig:factgraph}
\end{figure}
The propagation of LLR values during the decoding of polar code can be represented with an $m$-layer factor graph \cite{hussami2009perm}. An example of this graph is presented at Figure \ref{fig:factgraph}. Consider a permutation $\pi$ of its layers, or, equivalently, permutation $\hat \pi$ of the input LLR vector performed by applying $\pi$ to the binary representation of elements' positions. It can be seen that the bit estimation order also becomes permuted according to $\hat \pi$, and thereafter the subchannel reliabilities. 

It is easy to see that any factor graph layer permutation can also be expressed as action on monomials $\mathbf x\to P\mathbf x$ for $m\times m$ permutation matrix $P$.


\subsection{Projections}


The derivative in direction $b$ of a monomial is defined as
\[
D^b( x^v) = (x+b)^v - x^v.
\]

Similarly, the derivative of a monomial code $\mathcal C$ is the binary linear code with a generating set 

\begin{equation}
M_{\mathcal C\to b}=\left\{D^b(x^v) |x^v \in M_{\mathcal C}\right\}.
\label{eq:proj}
\end{equation}

Its codewords have the same values at positions $u$ and $u+b$, so we can write its generator matrix as $G=\begin{pmatrix}G^{(b)}&G^{(b)}\end{pmatrix}P'$, where $P'$ is a column permutation matrix and $G^{(b)}$ is a generator matrix of some ($2^{m-1},k^{(b)},d^{(b)}$) code $\mathcal C^{(b)}$, which will be subsequently denoted as the projected code or the projection. 

Observe that $(x+b)^v$ can be viewed as an action of the element of the translation group $\mathcal T_m$: $x\to x+b$ (which swaps the evaluation points $x$ and $x+b$). Then any codeword of the direction derivative is a sum of the codeword of $\mathcal C$ and its permutation. Therefore, a permutation defined by $b$ belongs to $\aut(\mathcal C)$ iff all elements of $M_{\mathcal C \to b}$ can be represented as sums of monomials from $ M_{\mathcal C}$ (one can also recall the definition of the weakly decreasing codes from \cite{bardet2016algebraic} and verify that $\mathcal T_m\in \aut(\mathcal C)$ is equivalent to the code being weakly decreasing).

Note that the derivative in the direction of unit vector $e_i$ is just a partial derivative:
\[
D^{e_i}(x^v) = \frac{\partial x^v}{\partial x_i},
\]
which implies that partial derivatives of monomial codes are also monomial. 

The projections are inherently connected to the Plotkin $(u|u+v)$ construction. Given a codeword of a monomial code $\mathcal C$, each projection corresponds to a vector $v\in \mathcal C^{(b)}$ and therefore performs the decomposition of $\mathcal C$ in $(u|u+v)$ form.

\section{Partially symmetric monomial codes}
For a BMS channel $W$, consider the decoding of some rate-$I(W)$ code $\mathcal C$ via its projections. The decoding of projected codes is performed in the 'XOR' channel (which corresponds to $W^-$ in polar coding notation). If a projection $\mathcal C^{(b)}$ has the rate $R_{\mathcal C^{(\mathbf b)}}$ that far exceeds the capacity of $W^-$, even the ML decoding is very likely to be erroneous and some techniques, e.g., list decoding, are needed to make use of that projection for the overall decoding. 

One way to think about it is through the polar coding and $(u|u+v)$ construction. When $R_{\mathcal C^{(b)}}$ is greater than $I(W^-)$, asymptotically there is a constant fraction of "badly" polarized data subchannels (in particular, $R_{\mathcal C^{(b)}}-I(W^-)$) and therefore the number of candidate values of $v$ that might correspond to an ML solution is exponential.

Polar codes are constructed so that for some $\hat i$ the projected code $\mathcal C^{(e_{\hat i})}$ has rate $\approx I(W^-)$. However, for all $ b\neq e_{\hat i}$ the projected codes have much higher rates, which makes the decoding via other projections inefficient. Any projection of Reed-Muller code RM$(r,m)$ is a Reed-Muller code RM$(r-1,m-1)$, so it is tempting to make different projections work together for the decoding purposes. However, these symmetries also constraint the decoding performance of any projection-based algorithm. Namely, if we consider a channel $W$ such that $I(W)=\dim \text{RM}(r,m)$, it is easy so see that the dimension of RM$(r-1,m-1)$ exceeds the capacity of $I(W^-)$. Furthermore, it is also possible to show that, e.g., if we fix the RM code rate to $1/2$, the rate of the projected code also converges to $1/2$ when the code length goes to infinity.

It is possible to establish the link between the symmetries of Reed-Muller codes and the (in)efficiency of projection-based decoding. By construction, all projections of RM codes are identical. As we will show further in this section, it puts a lower bound on the dimension of projections, and the RM codes, in fact, achieve this bound. 

One can now ask the question: if we sacrifice some of the code symmetries, what could be the potential gain? Would it be possible to achieve near-ML low-complexity decoding? In this paper, we investigate the properties of such codes. It turns out that the answer to the second question is mostly asymptotically negative. However, the considered framework and the corresponding codes can still be of sufficient interest in the finite-length regime and can serve as a foundation for further research. In this paper, we focus on partial derivatives due to their simplicity.

\begin{defn}
	A $t$-symmetric monomial code $\mathcal C_{m,t}$ is a binary monomial code of length $2^m$ such that the dimensions of $t$ of its partial derivatives are equal and the dimensions of $m-t$ others are strictly greater.
\end{defn}
In other words, there exists some set of \textit{target derivatives} $\mathcal H_t\subset \triv_m,|\mathcal H|=t$, such that $\forall e_i\in \mathcal H_t\ \dim \mathcal C^{(e_i)}_{m,t}=\tilde k_{m,t}$ and $\forall e_i\notin \mathcal H_t\ \dim \mathcal C^{(e_i)}_{m,t}>\tilde k_{m,t}$. A code is fully symmetric if $t=m$, non-symmetric if $t=1$ and partially symmetric otherwise. The Reed-Muller codes are fully symmetric and the polar codes in general are non-symmetric. Without loss of generality, we assume $\mathcal H_t=\{e_i|i\in \{1,\dots,t\}\}$.

As we established above, the efficiency of low-complexity decoding via projections depends on the rates of the projected codes. Therefore, for a given code dimension $k_{m,t}$ we are interested in the achievable values of $\tilde k_{m,t}$ and, in particular, the lower bound
\begin{equation}
\tilde k^*_{m,t}=\min_{\dim \mathcal C_{m,t}= k_m,{t}}\tilde k_{m,t}.
\end{equation}

\subsection{Lower bounds}
The problem of finding the minimum achievable projected code dimension $\tilde k^*_{m,t}$, or equivalently the minimum achievable rate $\tilde r^*_{m,t}$, for monomial codes is mostly combinatorial. 

Observe that when a monomial $x^v$ is removed from $M_{m}$, the dimension of derivative w.r.t. $x_i$ is decreased by 1 if $v_i=1$. Let us consider now the set $\mathcal M_t=\{x_1,\dots,x_{t}\}$. For $l\le t$, the number of monomials with $l\le t$ variables in $\mathcal M_t$ is $\binom{t}{l}2^{m-t}$ ($l$ variables out of $t$ can be selected in $\binom{t}{l}$ ways with any combination of the remaining $m-t$). We want the target derivatives to have equal dimensions, so for $l<t$ we have to remove $\lcm(t,l)/t$ monomials from $M_{m}$ at once and decrease the dimension of all $t$ target derivatives by $\lcm(t,l)/l$.

It is easy to see that the monomials with greater $l$ contribute more to the dimension reduction, hence in order to calculate the lower bound we at first remove the monomials that contain all $t$ variables, then the monomials that contain $t-1$, then $t-2$, etc. until the target $k_{m,t}$ is reached. Note that some values of $k_{m,t}$ cannot be achieved due to the granularity constraints. For $t=m$ and $k_{m,m}=\sum_{i=0}^{r}\binom{m}{i}$ we obtain the Reed-Muller code RM$(r,m)$.


Figure \ref{fig:r10} demonstrates the lower bound on $\tilde r^*_{m,t}$, computed for the case $m=10$ and $1\le t \le m$. This result can be interpreted as follows. If a curve lies above the reference BEC/BSC lines, the rate of the projected code exceeds the capacity of the underlying channel and the near-ML decoding via projections asymptotically needs exponential complexity. We can see that for most values of $t$ the efficient decoding can potentially be performed only for rates near 0 and 1. Case $t=3$ might be of some interest in BSC channel and $t=2$ can potentially be useful for any channels. Reed-Muller codes are on the $t=m$ curve, and for polar codes the bound $t=1$ is trivial since they lie exactly on the reference curves. An important observation here is that even though partially symmetric codes are mostly asymptotically bad, their rates scale better compared to RM codes and therefore they still need smaller decoding complexity in the finite-length regime.

\begin{figure}
	\centering
	\includegraphics[width=\linewidth]{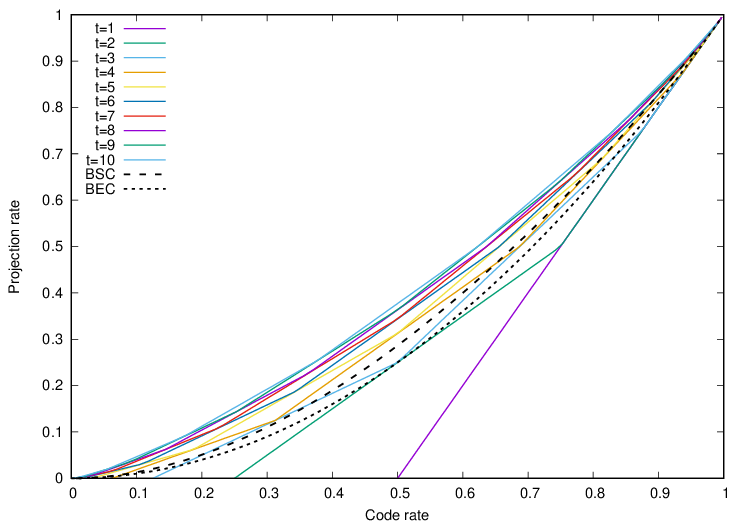}
	\caption{Achievable projected code rates, $n=1024$.}
	\label{fig:r10}
\end{figure}

We can use the procedure for obtaining the lower bound in order to establish the upper bound on the minimum distance of codes achieving $\tilde k^*_{m,t}$. For every $l$ we simply sort the monomials with $l$ variables in $\mathcal M_t$ by the total degree and perform the removal procedure from the highest to the lowest.
\begin{prop}
	\label{prop:dmin}
	Consider a code $\mathcal C_{m,t}$ which achieves the lower bound and assume that the monomial removal procedure stopped at some $l=\hat l$. Then the minimum distance of $\mathcal C_{m,t}$ is at most $2^{m-z}$, where $z=\hat l+m-t$ if less than $\binom{t}{\hat l}$ first entries at stage $l=\hat l$ were removed and $z=\hat l-1+m-t$ otherwise.
\end{prop}
This result follows directly from the construction. For any $l$, the maximal total degree of monomials considered in this stage is $l+m-t$ ($l$ variables from $\mathcal M_t$ and all $m-t$ from $\mathcal M_t^c$), and there are $\binom{t}{l}$ such monomials.

Unfortunately, it turns out that $t$-symmetric codes achieving the lower bound for small values of $t$ would have rather small minimum distance and therefore poor ML performance. For example, consider the construction of $2$-symmetric codes with dimension $k_{2}\ge 2^m-2^{m-2}$. The procedure stops at $\hat l=2$, so from proposition \ref{prop:dmin} we can immediately deduce that the minimum distance of such codes is at most 2. $3$-symmetric codes also have the minimum distance inferior to the polar codes. To overcome this issue, we suggest in practice to first remove all monomials up to a certain degree $r$, i.e., to construct $t$-symmetric codes $\mathcal C_{m,t}\subset \text{RM}(r,m)$.

\subsection{Code construction}
\label{ssDesign}
In this section we present the method for constructing codes that achieve $\tilde k^*_{m,t}$ for a given $k_{m,t}$. For the sake of simplicity, let us first assume $r=m$.
\begin{enumerate}
	\item Take $M_{\mathcal C}=M_{m}$, $k=|M_{m}|=2^m$.
	\item Set $\hat l=t$. While $k-2^{m-t}\binom{t}{\hat l} \ge k_{m,t}$, remove from $M_{\mathcal C}$ all monomials with $\hat l$ variables in $\mathcal M_t$ and decrease $\hat l$ by one, $k$ by $2^{m-t}\binom{t}{\hat l}$.
	\item Set $\hat d=m-t+\hat l$. While $k-\binom{t}{\hat l}\binom{m-t}{\hat d-\hat l} \ge k_{m,t}$, remove all degree-$\hat d$ monomials with $\hat l$ variables in $\mathcal M_t$ and decrease $\hat d$ by 1, $k$ by $\binom{t}{\hat l}\binom{m-t}{\hat d-\hat l}$.
\end{enumerate}
\begin{ex}
	Consider $m=4,t=3$ and $k_{4,3}=8$, $\mathcal M_3=\{x_1,x_2,x_3\}$. Table \ref{tab:monrem} contains all monomials with at least one variable in $\mathcal M_3$.
	\begin{table}[]
		\caption{Monomials to remove.}
		\label{tab:monrem}
		\centering
		\begin{tabular}{|c|c|c|}
			\hline
			\textbf{$l$} &\textbf{Impact on dimension} &\textbf{Monomials} \\ \hline
			$3$ & Remove 1 monomial  & $x_1x_2x_3x_4$,\\
			& $k$ decreases by 1 & $x_1x_2x_3$    \\ \hline
			$2$   & Remove 3 monomials  & $x_1x_2x_4$, $x_1x_3x_4$, $x_2x_3x_4$,   \\ 
			& $k$ decreases by 2&  $x_1x_2$, $x_1x_3$, $x_2x_3$,\\ \hline
			1 & Remove 3 monomials & $x_1x_4$, $x_2x_4$, $x_3x_4$\\ 
			& $k$ decreases by 1 & $x_1$, $x_2$, $x_3$\\ \hline
		\end{tabular}
	\end{table}
	
	Start from $M_{4}, k=16$ and go to step 2. Set $\hat l=3$. $16-2^{4-3}\binom{3}{3}=14\ge 8$, so we remove all monomials that contain $x_1x_2x_3$ ($x_1x_2x_3x_4$ and $x_1x_2x_3$), now $k=14$ and $\hat l=2$. $14-2^{4-3}\binom{3}{2}=8\ge 8$, so we remove all monomials that contain $x_1x_2$, $x_1x_3$ or $x_2x_3$  ($x_1x_2x_4$, $x_1x_3x_4$, $x_2x_3x_4$ and $x_1x_2$, $x_1x_3$, $x_2x_3$), now $k=8$ and the construction procedure is terminated since $k=k_{4,3}$. 
	
	The constructed $(16,8,4)$ code has generating set $M_{\mathcal C_{4,3}}=\{x_1x_4, x_2x_4, x_3x_4,x_1,x_2,x_3,x_4,1\}$. Its target directional derivatives have generating sets $M_{\mathcal C_{4,3}\to e_1}=M_{\mathcal C_{4,3}\to e_2}=M_{\mathcal C_{4,3}\to e_3}=\{x_4,1\}$ of cardinality $2$.
\end{ex}

For the case $r<m$, at step 1 we need to take $M_{r,m}$ instead of $M_{m}$, at step 2 term $2^{m-t}\binom{t}{\hat l}$ is replaced with $\binom{t}{\hat l}\sum_{i=0}^{\min(m-t,r-\hat l)} \binom{m-t}{i}$ 
and at step 3 the initial value of $\hat d$ becomes $\min(m-t+\hat l,r)$ (since after step 1 all monomials with degree greater than $r$ are already removed and therefore out of consideration).

\subsection{Achieving the remaining values of $k_{m,t}$}
\label{ssFin}
The procedure described in the previous section cannot obtain many achievable values of $k_{m,t}$. It stops when $k-\binom{t}{\hat l}\binom{m-t}{\hat d-\hat l}$ is smaller than $k_{m,t}$ for some $\hat l, \hat d, k\neq k_{m,t}$. So we need to remove $k-k_{m,t}$ degree-$\hat d$ monomials with $\hat l$ variables in $\mathcal M_t$ so that all $t$ target derivatives have the dimension $\tilde k^*_{m,t}$.

Consider the bipartite graph $\mathcal G=(V_L, V_R,E)$ with left vertices $u_{i}\in V_L$ isomorphic to variables $x_i,1\le i \le t$ and right vertices $u_{v}\in V_R$ isomorphic to all monomials $x^v$ of total degree $\hat d$, having $\hat l$ variables in $\mathcal M_t$. We draw an edge between $u_{i}$ and $u_{v}$ if the monomial $x^v$ contains $x_i$. This graph is $(\binom{t-1}{\hat l-1}\binom{m-t}{\hat d-\hat l},\hat l)$-biregular and we want to remove $k-k_{m,t}$ its right vertices so that the graph remains biregular, i.e., find its $(\cdot, \hat l)$-biregular subgraph $\mathcal G'$. Without loss of generation assume $t=m$ (otherwise we can split $\mathcal G$ into partitions corresponding to different monomials from $\mathcal M_t^c$). Figure \ref{fig:subgraph} demonstrates an example of graph $\mathcal G$ for $m=4,\hat d=\hat l = 2$ and one of its possible $(1,2)$-regular subgraphs $\mathcal G'$ (in red).

\begin{figure}
	\centering
	\scalebox{1}{\begin{tikzpicture}[thick,scale=0.55, every node/.style={scale=0.6}]
\tikzstyle{lnode}  = [radius=0.5, black, thick, fill=white];
\tikzstyle{rnodeS}  = [radius=0.5, red, thick, fill=white];
\tikzstyle{lineG}  = [-, thick,  red];

\coordinate (0) at (1,1);
\coordinate (1) at (1,3);
\coordinate (2) at (1,5);
\coordinate (3) at (1,7);
\coordinate (4) at (1,3);
\coordinate (20) at (1.5,1);
\coordinate (21) at (1.5,3);
\coordinate (22) at (1.5,5);
\coordinate (23) at (1.5,7);
\coordinate (24) at (1.5,3);

\coordinate (5) at (6,-1);
\coordinate (6) at (6,1);
\coordinate (7) at (6,3);
\coordinate (8) at (6,5);
\coordinate (9) at (6,7);
\coordinate (10) at (6,9);
\coordinate (35) at (5.5,-1);
\coordinate (36) at (5.5,1);
\coordinate (37) at (5.5,3);
\coordinate (38) at (5.5,5);
\coordinate (39) at (5.5,7);
\coordinate (310) at (5.5,9);
\draw (0) [rnodeS] circle;
\draw (0) node {$x_4$};
\draw (1) [rnodeS] circle;
\draw (1) node {$x_3$};
\draw (2) [rnodeS] circle;
\draw (2) node {$x_2$};
\draw (3) [rnodeS] circle;
\draw (3) node {$x_1$};

\draw (5) [rnodeS] circle (0.6cm);
\draw (5) node {$x_3x_4$};
\draw (6) circle (0.6cm);
\draw (6) node {$x_2x_4$};
\draw (7) circle (0.6cm);
\draw (7) node {$x_2x_3$};
\draw (8) circle (0.6cm);
\draw (8) node {$x_1x_4$};
\draw (9) circle (0.6cm);
\draw (9) node {$x_1x_3$};
\draw (10) [rnodeS] circle (0.6cm);
\draw (10) node {$x_1x_2$};

\draw  [lineG] (20) -- (35);
\draw (20) -- (36);
\draw (20) -- (38);
\draw  [lineG] (21) -- (35);
\draw (21) -- (37);
\draw (21) -- (39);
\draw (22) -- (36);
\draw (22) -- (37);
\draw  [lineG] (22) -- (310);
\draw (23) -- (38);
\draw (23) -- (39);
\draw  [lineG] (23) -- (310);

\end{tikzpicture}}
	\caption{$(3,2)$-regular $\mathcal G$ and $(1,2)$-regular $\mathcal G'$}
	\label{fig:subgraph}
\end{figure}
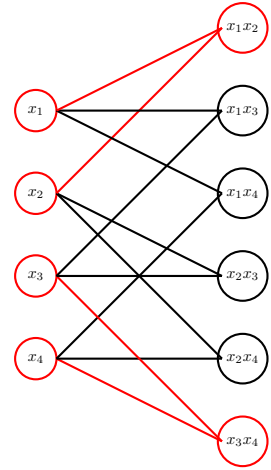

For every achievable $k_{m,t}$, it can be seen that the biadjacency matrix of $\mathcal G'$ satisfies the existence conditions from \cite{brualdi1980matrices}, so we can always find such $\mathcal G'$. This can be done, e.g., by finding the maximum flow in a network with the source connected to all left vertices with capacity-$\frac{\lcm(m, \hat l)}{\hat l}$ edges, the sink connected to all right vertices with capacity-$m$ edges and all $e\in E$ having the capacity 1.

\section{Structure of partially symmetric codes}
Let us define $\check{\mathcal C}_{m,t}$ as the code obtained using the procedure from Section \ref{ssDesign}.

\begin{prop}
	For any $\check{\mathcal C}_{m,t}$ holds $\mathcal T_m\in \aut(\check{\mathcal C}_{m,t})$.
\end{prop}
This property follows directly from the code construction. Consider a monomial $x^v\in M_{\check{\mathcal C}_{m,t}}$ and assume it has $\tilde l$ variables in $\mathcal M_t$. Any divisor of $x^v$ has $\le \tilde l$ variables in $\mathcal M_t$ and smaller degree, so it cannot be removed from $M_{\check{\mathcal C}_{m,t}}$ nefore $x^v$ in the construction process. As a consequence \cite{bardet2016algebraic}, the dual code $\check{\mathcal C}_{m,t}^\bot$ is also monomial and $\mathcal T_m\in \aut(\check{\mathcal C}_{m,t}^\bot)$.

\begin{prop}
	$\check{\mathcal C}_{m,t}$ is invariant w.r.t. any permutation on sets $\{x_1,\dots,x_{t}\}$ and $\{x_{t+1},\dots,x_{m}\}$.
\end{prop}
It is sufficient to observe that for any $d,l$ all degree-$d$ monomials with $l$ variables in $\mathcal M_t$ (and $d-l$ variables in $\mathcal M_t^c$) are either in $M_{\check{\mathcal C}_{m,t}}$ or its complement.

Note that in some cases codes $\check{\mathcal C}_{m,t}$ might coincide with the ones from \cite{kamenev2019polar} (which are for some choice of parameters invariant w.r.t. $\{x_1,\dots,x_{t}\}$, but are optimized for SC decoding error probability rather than the projected code dimensions), but contrary to them the construction of $\check{\mathcal C}_{m,t}$ is channel-independent.

\begin{prop}
	For any $e_i\in \mathcal H_t$ code $\check{\mathcal C}_{m,t}^{(e_i)}$ achieves $\tilde k^*_{m-1,t-1}$.
\end{prop}
Consider the generating set of code $\check{\mathcal C}_{m,t}$. It contains all monomials of two types:
\begin{enumerate}
	\item With less than $\hat l$ variables in $\mathcal M_t$;
	\item With exactly $\hat l$ variables in $\mathcal M_t$ and total degree less than $\hat d$
\end{enumerate}
for some $\hat l, \hat d$. When we take the derivative in the direction $e_i\in \mathcal H_t$, the monomials of first type now have less than $\hat l-1$ variables in $\mathcal M_t\setminus \{x_i\}$ and the monomials of second type now have exactly $\hat l-1$ variables in $\mathcal M_t\setminus \{x_i\}$ and total degree less than $\hat d-1$. The only thing left here is to notice that this matches exactly the description of code $\check{\mathcal C}_{m-1,t-1}$.

\section{Numeric results}
We demonstrate the performance of partially symmetric monomial codes in BEC($\varepsilon$), where the polynomial-time ML decoding is available. Figure \ref{fig:ferml512} shows the frame error rate of codes of length $512$ and rate $1/2$, constructed for different values of $t$ such that the minimum distance of the constructed codes for $t\le 7$ is equal to 16 (for $t=8=9$ we get the Reed-Muller code with $d_{min}=32$). The polar code at the figure is constructed for each value of $\varepsilon$. Its minimum distance depends on the target erasure probability and jumps from $8$ to $16$ between $\varepsilon=0.36$ and $\varepsilon=0.38$. 

\begin{figure}
	\centering
	\includegraphics[width=\linewidth]{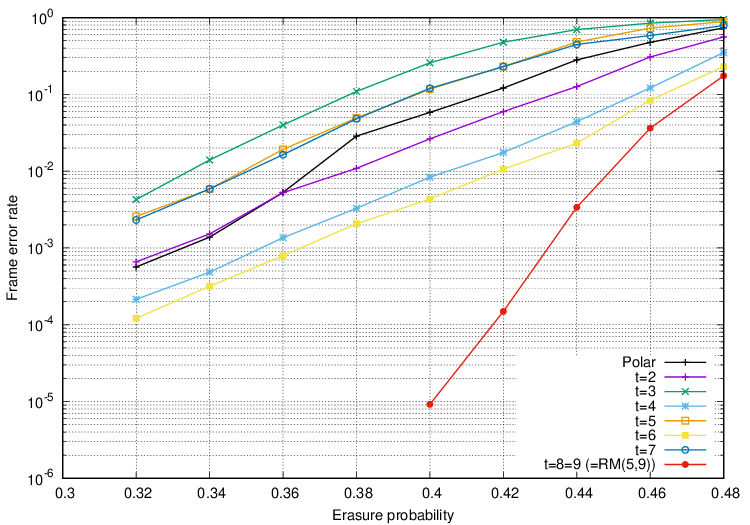}
	\caption{ML performance in BEC, $n=512, k=256$.}
	\label{fig:ferml512}
\end{figure}

One conclusion that can be made from the picture is that the ML performance does not strictly improve with $t$. However, partially symmetric codes for some values of $t$ demonstrate better performance compared to the polar code, although the gap is not large.

To illustrate the projection-based decoding performance, let us consider a variation of projection-aggregation (PA) algorithm \cite{ye2019recursive} for the case of BEC($\varepsilon$). For a code $\mathcal C$, consider some set $\mathcal B$ of its projections. Then in order to decode the erased codeword $(c_0, \dots, c_{2^m-1})$ of $\mathcal C$, we compute the projected codewords $c^{(b)}$ and perform the bitwise MAP decoding on them. Any recovered position of $c^{(b)}$ can be written as $c_{u}\oplus c_{u+b}$, so if only one of the positions $c_{u},c_{u+b}$ is erased, we can recover another one. The process is iteratively repeated until all erasures are recovered or no position was recovered in the current iteration. 

Figure \ref{fig:fer512} shows the performance of partially symmetric codes $\mathcal C_{9,t}\subset \text{RM}(3,9)$, as well as polar and Reed-Muller codes. Set $\mathcal B$ here consists of 16 (out of 511 total) projections with minimum dimensions (which are more likely to be decoded successfully). For polar, $\mathcal C_{9,2}$ and $\mathcal C_{9,4}$ codes PA decoding achieves the maximum likelihood performance, so it is not given on the plot. The polar code given at the figure is constructed for each value of $\varepsilon$. Its minimum distance depends on the target erasure probability and jumps from $8$ to $16$ between $\varepsilon=0.36$ and $\varepsilon=0.38$, which can be observed at the figure. 

\begin{figure}
	\centering
	\includegraphics[width=\linewidth]{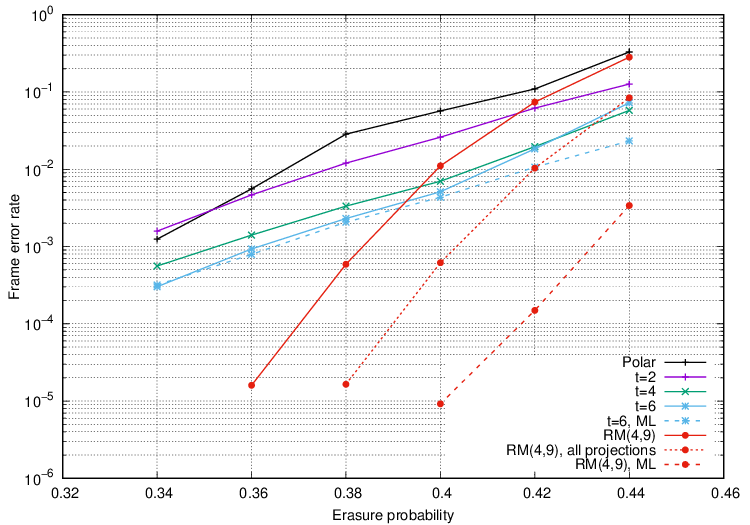}
	\caption{Projection-based decoding in BEC, $n=512, k=256$.}
	\label{fig:fer512}
\end{figure}

It can be seen that for small values of $t$ the projection-based decoding can achieve near-ML performance. However, the projections' dimensions grow with $t$, so the efficiency of this approach drops down. This is particularly noticeable for the larger values of $\varepsilon$. For the Reed-Muller codes (which are fully symmetric), ML performance cannot be reached even when with all $511$ projections. Note that the presented partially symmetric codes demonstrate the performance superior to the polar code, although the gap is not large.

As for the future work, one might think of how to construct the partially symmetric polynomial codes. It is known that non-symmetric polynomial codes can bring significant performance boost under polar-like decoding with a small list size \cite{trifonov2016subcodes}. One can wonder whether a similar improvement can be achieved with (partially) symmetric codes and what are the limitations.

\section{Conclusion}
In this paper, we introduce the new family of monomial codes. These codes have a smaller group of symmetries compared to RM codes but are more adapted for the projection-based decoding. A construction of such codes is proposed and it is shown that the obtained codes in some cases demonstrate the recursive structure.


\bibliographystyle{IEEEtran}
\bibliography{template}
\end{document}